\documentclass[10pt,conference]{IEEEtran}
\usepackage[cmex10]{amsmath}
\usepackage{amsthm}
\usepackage{amssymb}
\usepackage{array}
\newtheorem{theorem}{Theorem}
\newtheorem{proposition}{Proposition}
\newtheorem{lemma}{Lemma}
\newtheorem{definition}{Definition}
\newtheorem{remark}{Remark}

\usepackage[linesnumbered,ruled,vlined]{algorithm2e}
\usepackage{graphicx}
\usepackage{epstopdf}
\DeclareGraphicsExtensions{.pdf,.jpeg,.png,.eps}

\newcommand{\argmax}{\arg\!\max}
\IEEEoverridecommandlockouts
\begin{document}

\title{A Distributed Satisfactory Content Delivery Scheme for QoS Provisioning in Delay Tolerant Networks}
\author{\IEEEauthorblockN{Sidi Ahmed Ezzahidi}
\IEEEauthorblockA{LIMIARF,\\ University of Mohammed V, \\B.P. 1014 RP, Rabat, Morocco\\
sa.ezzahidi@gmail.com}
\and
\IEEEauthorblockN{Essaid Sabir}
\IEEEauthorblockA{UBICOM Research Group,\\ ENSEM, Hassan II \\University of Casablanca, Morocco.\\
e.sabir@ensem.ac.ma}
\and
\IEEEauthorblockN{ Mounir Ghogho}
\IEEEauthorblockA{School of Electronic and Electrical Engineering,\\ University of Leeds, \\United Kingdom\\
m.ghogho@leeds.ac.uk}
}

\maketitle
\begin{abstract}
We deal in this paper with the content forwarding problem in Delay Tolerant Networks (DTNs). We first formulate the content delivery interaction as a non-cooperative satisfaction game. On one hand, the source node seeks to ensure a delivery probability above some given threshold. On the other hand, the relay nodes seek to maximize their own payoffs. The source node offers a reward (virtual coins) to the relay which  caches and forwards the file to the final destination. Each relay faces the dilemma of accepting/rejecting to cache the source's file. Cooperation incurs energy cost due to caching, carrying and forwarding the source's file. Yet, when a relay accepts to cooperate, it may receive some reward if it succeeds to be the first relay to forward the content to the destination. Otherwise, the relay may receive some penalty in the form of a constant regret; the latter parameter is introduced to make incentive for cooperation. Next, we introduce the concept of Satisfaction Equilibrium (SE) as a solution concept to the induced game. Now, the source node is solely interested in reaching a file delivery probability greater than some given threshold, while the relays behave rationally to maximize their respective payoffs. Full characterizations of the SEs for both pure and mixed strategies are derived. Furthermore, we propose two learning algorithms allowing the players (source/relays) to reach the SE strategies. Finally, extensive numerical investigations and some learning simulations are carried out to illustrate the behaviour of the interacting nodes.
\end{abstract}

\section{Introduction}
Nowadays, self-organizing is tremendously becoming a key feature for current and future mobile networking. 
Moreover, numerous new applications and some special circumstances require the nodes/network to be self-organizing, self-configuring and self-healing. In order to overcome extreme circumstances (earthquakes, disasters, ...), massive access to the network (sport events, festival, ...) or lack of infrastructure in general, a class of self-organizing networks called Delay Tolerant Networks (DTN) \cite{fall2003delay} has been proposed and are continuously gaining interest. A DTN is a class of an infrastructure-less and fully distributed wireless networks. Such a network is designed to operate over arbitrary distances, including very small scale (e.g., cells communications) to ultra large scale (e.g., interplanetary communications). Intermittent connectivity and an excessively large delay may occur very often in such an environment, which makes the end-to-end connectivity a very challenging issue. Thus, the use of store-carry-forward paradigm seems to be an attractive solution. The main idea here is to exploit the opportunistic inter-contacts between relay nodes to cache-and-forward given data to the final destination.\\

\par DTNs-based applications are very various \cite{gao2015delay}, including digital communication for rural areas (e.g., DarkNet, TrainNet, KioskNet, etc.), personal/wildlife communications (e.g., Pollen, Body Area Networks, ZebraNet, etc.), battlefield communications (e.g., Military Missions and Airborne Networks) disaster rescues and environmental monitoring communications. Moreover, delay tolerant networks could be an attractive/efficient solution to offload legacy networks. They may help out to control the congestion caused by the exponential growth of traffic, heterogeneity in infrastructure/technologies (Device-to-Device, RFID, Drone-based backhaul, etc.), see Fig.\ref{5g}. Moreover, they are expected to be a part of the next generation radio communication system such as 4G LTE-Advanced networks, 5G and Internet of Things (IoT) \cite{prasad2015efficient, 10.4108/eai.26-10-2015.150598}. We recall also the ``Any Time, Any Where, Any Device'' (ATAWAD) paradigm which has been fueled by the prevalence of devices, exploiting  collaboration between devices, enhancing functionalities and speeding up the access to data. Technically, DTNs can be helpful to assist data transmission to/from isolated regions where connectivity is hard or even impossible to establish. Moreover, the DTNs devices could be efficient to enhance Quality of Service (QoS) and to reduce energy consumption by strategically offloading the traffic from the network backbone. However, getting nodes in the network to cooperate and act as relays is a still fundamental challenge, due to limited bandwidth, limited battery energy and limited storage capacity. Indeed, participating in data transmission and content caching incurs energy consumption.  Therefore, the relays may exhibit a selfish behavior, which significantly degrades the DTNs performance. The latter is a key motivation to develop and deploy efficient and distributed mechanisms to solve the inherent cooperation issue.\\

\begin{figure}[!htbp]
\centerline{\includegraphics[height=5cm,width=9cm]{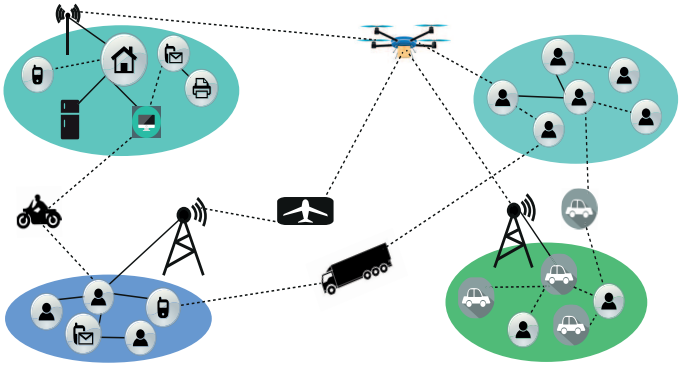}}
\caption{\footnotesize{DTN communications over Device-to-Device communications.}}\label{5g}
\end{figure} 

\par
In order to evaluate the network performance, several metrics can be utilized. For instance, one can use the delivery rate, the content loss rate, the protocols overhead, the end-to-end delay, the expected number of transmissions and the energy consumed. Here we consider optimizing the energy consumption subject to the average delivery rate being above some given threshold. In other words, this work is devoted to presenting a new fully distributed framework for QoS support in DTN-like networks.\\

\par 
Self-configuring capability and the distributed nature of DTNs are proven to induce selfish behaviour at nodes level, see \cite{karaliopoulos2009assessing,li2010routing,li2010evaluating}. Several incentive mechanisms have been designed to sustain cooperation among selfish relay nodes. Game theory seems to be the perfect tool to design  such mechanisms. It is mainly viewed as a tool to investigate the decision-making by the system through an equilibrium analysis instead of optimality analysis. The authors in \cite{jiang2013survey} presented an interesting survey on incentive mechanisms for DTNs. They compared four schemes: 1) virtual currency based incentive mechanism, 2) credit-based incentive mechanism, 3) game-theory-based incentive mechanism and 4) combined incentive mechanism. In \cite{hulke2015game}, the authors discussed different game theoretical-inspired incentive mechanisms and analyzed them while pointing out their advantages and drawbacks. Interestingly, many works have proposed curious schemes to encourage nodes to cooperate. For example, the authors in \cite{el2013evolutionary} used evolutionary games theory and addressed how a reward mechanism could efficiently induce cooperation among relay nodes in delay tolerant networks. The authors in \cite{brun2014modeling} suggested a simple reward-based mechanism scheme, and show how the source could optimally set the reward value based on the relays information.\\

\par
In \cite{long2007non}, a repeated game is constructed in order to capture the interaction between mobile nodes in terms of non-cooperative power control. The authors in \cite{perlaza2012quality} introduced the framework of satisfaction form of game to model the problem of QoS provisioning in decentralized networks, and they provided a comparison between the concept of Generalized Nash Equilibrium (GNE) and the concept of Satisfaction Equilibrium (SE). To the best of our knowledge, this is the first work conducted to analyze distributed caching in DTNs under quality of service constraints using satisfaction game approach.\\

\par
The main  contribution  of  this  paper is the design of a new distributed framework for QoS provisioning in DTN-like networks using the concept of satisfaction equilibrium and a reward-based incentive mechanism. By offering some reward (virtual coins) to relay nodes, the source node can efficiently encourage them to use a part of their battery energy  and  participate to forward a file to some given destination, i.e, caching a given file and wait to find a persistent connection with the interested destination. Within this framework, the resulting distributed caching problem is investigated using the powerful tool of non-cooperative game theory. Notice that in this paper we use the concept of SE \cite{tian1992existence} as a solution concept instead of the well-known Nash equilibrium concept. Now, one needs not only a stable state of the game but also  providing certain performance requirements. More precisely, the source has an specific problem which is offering a minimum reward ensuring that the relays delivery probability does not go bellow a threshold value, while the relays have the choice to "accept" or "reject" this offer, depending on the reward value whether it is beneficial or not. Next, we exhibit sufficient conditions for existence of an SE for both pure strategies and mixed strategies. Moreover, aiming to understand the behavior of the DTN source-relay nodes during the interaction and the eventual convergence to the SE, we propose two stochastic algorithms for both the source and the relay nodes. \\

\par
The remaining sections of the paper are organized as follows. In section \ref{MDPF}, we describe the problem, its formulation and our solution design. In Section \ref{GTM}, we present the game theoretical model including, utility functions and the SE formal definition. Section \ref{SEA} exhibits the satisfaction equilibria computation and an analysis of their efficiency. Next, we describe the stochastic learning algorithms adopted in Section \ref{LA}. Section \ref{MI} provides some numerical investigations and simulation runs to claim our work and a conclusion is drawn in Section \ref{Con}.

\section{Problem Formulation}\label{MDPF}

We consider a delay tolerant network including a pair of source-destination, and $n$ relay nodes. When a contact between the source node and a relay node takes place, the source transmits a data file to the relay node. Relay nodes, moving independently in the network area, store the file, carry it and wait until having a direct link opportunity with the destination node to forward the file. We next list the assumptions considered in this paper:
\begin{itemize}
\item The file to forward has a finite lifetime $\tau$ (called also horizon) during which the destination is interested in its content;
\item For simplicity and without loss of generality, we assume that  all nodes are identical and equipped with the same wireless interface;
\item We consider that the relay nodes use a two-hop routing policy \cite{panagakis2007study}, which works as follows: when a relay node receives a copy of the file from the source, it stores it and forwards it to the destination node when met within the file's lifetime. The choice of this routing protocol is motivated by an energy efficiency purpose. Indeed such a routing scheme has a good delivery/energy efficiency trade-off;
\item Occurrence of the contacts between any two nodes follows Poisson distribution. Thus, the time interval between two successive contacts (inter-contact time) for each pair of relay nodes is exponentially distributed with a pairwise meeting rate $\lambda \geq 0$. A comprehensive discussion of this modeling can be found in \cite{chaintreau2007impact} and \cite{groenevelt2005message}. We further consider that the contact time is large enough for the complete transmission of the source's file.
\end{itemize} 
The source-relay contact probability within the file's life time $\tau$ can be expressed as,
 \begin{eqnarray}
p_c =p( t \leq \tau)= \int_0^{\tau}  \lambda e^{-\lambda t} dt= 1-e^{-\lambda \tau}.
\end{eqnarray}
Each relar node can be in one of the three states: 1) listening (sending beacons for discovery purposes), 2) transmitting (when in contact with the destination), or 3) receiving data from the source node. In this paper, we neglect the energy consumption related to the discovery/listening state. Thus, the energy consumption per node includes the energy consumed during the receiving state (denoted $e_r$), the energy consumed while transmitting the file to the destination (denoted $e_t$) and the storing energy per time slot (denoted $e$). It follows that the mean energy $e_s$ dissipated while caching a file with lifetime $\tau$ can be  written as
\begin{eqnarray}
e_s &=& \int_0^{\tau} e \lambda t e^{-\lambda t} dt\nonumber \\ 
	&=& e\left[\frac{1- (1+\lambda \tau) e^{-\lambda \tau}}{\lambda}\right] \nonumber \\
	&=& \frac{e}{\lambda}\left(1-Q_{\tau}\right).
\end{eqnarray}
where $Q_{\tau}$ denotes the probability that a given relay fails in relaying a given file to the destination \cite{altman2009competition}. Hence, the total energy consumption per node becomes $\eta =e_r+e_t+e_s$.

\subsection{Problem Formulation and Solution Design}\label{DOCP}
 
 It is desirable that the nodes behave in a fully cooperative fashion in order to maximize the overall delivery rate. However, in real deployments, the relay nodes may not cooperate due to energy constraints and conflicting interests. Hence, a degradation of the network performance may be observed. To deal with this problem and encourage/force the relay nodes to participate in the file forwarding game, we next develop a reward-based incentive mechanism.\\
 \begin{figure}[!htbp]
\centerline{\includegraphics[height=3.5cm,width=8cm]{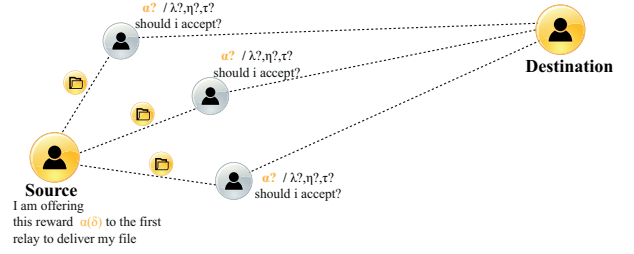}}
\caption{\footnotesize{The Source-Relay interaction.}}\label{Reward}
\end{figure} 

\par Fig.\ref{Reward} describes the interaction between the source node and the relay nodes under the proposed rewarding mechanism. On one hand, the source has the objective of making its file arrive to destination during the file's lifetime, so it generates a copy of this file and attempts to convince the relays encountered to forward it. On the other hand, the encountered relays can either accept (strategy `a') or reject (strategy `r') caching and relaying this file. In order for the long term average delivery rate of the source files to be above some given threshold (minimum QoS), the source needs to determine the appropriate reward to offer to the cooperating relays. When a source-relay contact occurs, the rational relay nodes seek to pick a strategy that maximizes its own payoff. In the meantime, the source offers the reward value merely to satisfy an individual constraint which consists of reaching  a delivery rate higher that some threshold $\delta$. Moreover, when a  relay node declines the forwarding offer or fails to reach the destination within the file's lifetime, it incurs a penalty in the form of a regret of declining or failing. In the next section, we construct a game theoretic framework to capture the performance of the proposed reward-based mechanism.

\section{Game Theoretical Model}\label{GTM}
Game theory has been used to solve problems in ad hoc, fixed and cellular networks. It is a powerful tool for the analysis of distributed networks. Its equilibrium concept and formulation of the utility function under constraints permit to study the system behavior and its decision strategies. It is mainly used to study the decision-making by the system through an equilibrium analysis. Indeed, in this section we first formulate our case of DTN as a homogeneous One-Shot caching Game, where the source and relays are selfish players playing independently and simultaneously. The source's strategy is the choice of the continuous-valued reward $\alpha$, taken from the interval $[0,\alpha_{max}]$,  and the relays have two discrete strategies accept 'a' or reject 'r'. Then, we present the utility function which considers the proposed reward mechanism, the energy consumption, the delivery probability and the regret values. Next, we study the existence and uniqueness of equilibria in pure game when the players choose to play their pure strategies and in mixed game when they independently and randomly select their strategies; in the mixed game, each player (i.e. each relay in contact with the source) accepts with probability $p \in [0,1]$ or reject  with $1-p$. Later, we give the conditions driving the system to an operating/stable point, namely a Satisfaction Equilibrium \cite{nash1950equilibrium}. \par
In fact, according to the game's concept, each  player's action/strategy corresponds to a certain utility. Rationally, when a relay receives a file from the source, its utility depends  on the choice of other relays' strategies. As mentioned above, the relays 'meet' the source with probability $p_c$ and all arrivals are independent. Therefore, the average number of relays among $n$ which are in contact with the source is $\tilde{n}=np_c$.
\par
Further, according to  \cite{altman2009competition}, we define $\phi^i(\tilde{n})$ the delivery probability that a given relay $i$ among $\tilde{n}$, plays pure strategy 'a' and succeeds to deliver a given file ( i.e., the first to deliver the file to the destination) as 
\begin{eqnarray}
\label{Psucc}
   \phi^i(\tilde{n})&=& (1-q_{\tau}) \sum_{j=1}^{\tilde{n}} {\tilde{n}-1 \choose j-1} \frac{(1-q_{\tau})^{\tilde{n}-1} q_{\tau}^{\tilde{n}-j}}{j}\nonumber\\&=&\frac{1-q_{\tau}^{\tilde{n}} }{\tilde{n}}.
\end{eqnarray}
Next, we will formulate the utility function for each player.
\subsection{Utilities Formulation}
We define the utility function of each relay as the difference between the reward that it can win from the source and the energy consumed to cooperate with the source. Thus, we denote $U_i(\mbox{'a'},\alpha,,n_a)$ the utility function of  a given relay when it plays its pure strategy accept 'a', and  $U_i(\mbox{'r'},\alpha,,n_a)$ when it plays its pure strategy reject 'r'.
\begin{equation}
  \label{ru}
\begin{cases}
& \mathcal{U}_i(\mbox{'a'},\alpha,n_a)= \alpha \phi^i(n_a) -\sigma(1- \phi^i(n_a))-\eta\\ \\
 &\mathcal{U}_i(\mbox{'r'},\alpha,n_a)=-\alpha  \phi^i(n_a)-\gamma \\
\end{cases}
\end{equation}\par
where  $\sigma$ and $\gamma$ are appropriate constants (regrets) that the relay incurs when it accepts to cache a given file but it does not succeed to deliver it during its lifetime and when it declines the source offer respectively. $n_a$ denotes the number of relays out of $\tilde{n}$, that are have been in contact with the source and that accepted to cache .\par
Furthermore, because of the selfish behavior, each relay decision is based on self-optimisation; the relay's  objective is then to maximize its own utility in a distributed fashion, i.e.
\begin{equation*}
\max_{v \in \{a,r\}}  \mathcal{U}_i(\mbox{'v'},\alpha,n_a), \quad \forall i =0, 1, \cdots, n_a.
\end{equation*}\par
The utility function of source $\mathcal{U}_s(\mbox{'a'},,n_a)$  is merely the delivery probability of  all relays that have accepted caching.
\begin{equation*}
 \mathcal{U}_s(\mbox{'a'},n_a)=n_a \phi^i(n_a), \quad \forall i =0, 1, \cdots, n_a.
\end{equation*}

\subsection{Satisfaction Equilibrium}
Now, we present a definition for the Nash Equilibrium \cite{nash1950equilibrium}, which is the point where no player can improve his payoff by making individual changes in his decisions. Precisely, in our source-relays game we introduce an equilibrium called satisfactory equilibrium \cite{tian1992existence}: the source seeks to satisfy the constraint that its 'well-being' should be greater than a fixed value (threshold), given the strategies adopted by the relays, which behave rationally to maximize their ``well-being'' by playing their optimal strategies.
\begin{definition}
 At satisfaction Nash  equilibrium $(p^*,\alpha^*)$ we have, 
\begin{equation*}
 E[\mathcal{U}_s(\mbox{'a'},n)] \geq \delta, \quad \forall \alpha \in [0,\alpha_{max}],
\end{equation*}
and
\begin{equation*}
\displaystyle p^* \in \argmax_{p \in [0,1]}E[ \mathcal{U}_i (p,\alpha^*,n)], \quad  \forall i =0, 1, \cdots, n_a.
 \end{equation*}
\end{definition}\par
Next, we will perform a thorough analysis of SE, its existence and uniqueness in pure an mixed strategies. Moreover, we discuss the full conditions that drive our distributed system to this steady point, if there exists one. 

\section{Satisfaction equilibrium analysis}\label{SEA}
We analyze here the structure of the content caching game solutions and we derive many sufficient conditions for the existence of a SE.

\begin{theorem}
(Nash's Theorem \cite{nash1950equilibrium}) Every finite game in strategic form (i.e., with finite number of players and finite number of pure strategies for each player) has at least one Nash
equilibrium (NE) (involving pure or mixed strategies).
\end{theorem}
Since the condition $\alpha(\delta)$ is imposed on the  game, the NE is not a suitable solution. Along this paper we replace the Nash equilibrium by the Satisfaction Equilibrium which is more natural.

\begin{remark}
We highlight  that the existence of  pure / mixed  SE does not necessarily imply its  uniqueness. In fact, the existence of the constraint $\alpha(\delta)$ yields the fact to provide conditions to have a unique SE very difficult.
\end{remark}

\subsection{Pure Satisfaction Equilibria (PSE)}
We turn now to derive the Satisfaction Equilibria for pure strategy case. The players act with their pure strategies.

\begin{lemma}The content caching game may have numerous PSE. Satisfaction Equilibria are any ($\alpha^*$, $n_a ^*$) solutions of the following two conditions:
\begin{equation}
n_a ^*\geq \frac{\log(1-\delta)}{log(q_{\tau})}
\end{equation}
and,
\begin{align}
\alpha^*=\frac{\lambda\sigma(n-1+q_{\tau}^{n_a})-n_a(\lambda(\gamma-e_r-e_t)-e (1-q_{\tau}))}{2\lambda(1-q_{\tau}^{n_a})}.
\end{align}
\end{lemma}
\begin{proof}
Assume that the profile $ \mathcal{J}=(\overbrace{\mbox{'a'},\mbox{'a'},...,\mbox{'a'}}^{n_a},\overbrace{\mbox{'r'},\mbox{'r'}....,\mbox{'r'}}^{n-n_a})$ is a satisfaction equilibrium. Then
\begin{equation*}
\begin{cases}
& \mathcal{U}_i(\mbox{'a'},\alpha,n_a) \geq \mathcal{U}_i(\mbox{'r'},\alpha,n_a) ,\quad  \forall i =0, 1, \cdots, n_a \\ \\
 &\mathcal{U}_i(\mbox{'a'},\alpha,n_a) \leq \mathcal{U}_i(\mbox{'r'},\alpha,n_a),  \quad \forall i =0, 1, \cdots, n-n_a \\
\end{cases}
\end{equation*}
Hence, the relays  are indifferent between the two strategies \mbox{'a'}, \mbox{'r'}, then 
\begin{equation*}
U_i(\mbox{'a'},\alpha^*,n_a) = U_i(\mbox{'r'},\alpha^*,n_a),
\end{equation*}
\begin{equation*}
\phi^i(n_a) -\sigma(1- \phi^i(n_a))-\eta= -\alpha  \phi^i(n_a)-\gamma.
\end{equation*}
After a few lines of algebra, we find
\begin{align*}
\alpha^*=\frac{\lambda\sigma(n-1+q_{\tau}^{n_a})-n_a(\lambda(\gamma-e_r-e_t)-e (1-q_{\tau}))}{2\lambda(1-q_{\tau}^{n_a})}.
\end{align*}
The source problem consists of assuring a succeed relays' probability  value
\begin{equation*}
 n_a\phi^i(n_a) \geq \delta,
\end{equation*}
then,
\begin{equation*}
1- q_{\tau}^{n_a} \geq \delta, \quad \Longrightarrow  \quad  n_a ^*\geq \frac{\log(1-\delta)}{log(q_{\tau})}
\end{equation*}
\end{proof}\par
The pure equilibrium could fail to achieve a certain lucidity between relay nodes since only a part of relays may accept to cache the file. To solve this problem, we use another concept of equilibrium, named Mixed Satisfaction Equilibrium, in which the relay will accept to cache the file with some probability.

\subsection{Mixed Satisfaction Equilibria (MSE)}
When mixed strategy is allowed, the relays randomize between accepting and rejecting the source offer according to common probability distribution, accepting with $p$, rejecting  with  $1-p$.
\begin{lemma}The satisfactory caching game has infinitely many Mixed Satisfaction Equilibria ($\alpha^*$, $p^*$). They are solutions of the following system:
\begin{equation}
\begin{cases}
& \alpha^*=\frac{\lambda\sigma(n-1+(1-(1-q_{\tau})p_cp^*)^{n})-n(\lambda(\gamma-e_r-e_t)-e (1-q_{\tau}))}{2\lambda(1-(1-(1-q_{\tau})p_cp^*)^{n})}.\\  \\
&   p^* \geq \frac{1- \sqrt[n]{1- \delta}}{(1-q_{\tau})p_c} \\
\end{cases}
\end{equation}
\end{lemma}
\begin{proof}
At the equilibrium each relay is indifferent about which strategy to choose. Namely
\begin{equation*}
\mathcal{U}_i(\mbox{'a'},\alpha^*,n)=\mathcal{U}_i(\mbox{'r'},\alpha^*,n),
\end{equation*}
\begin{equation}
\label{eq}
\alpha^* \phi^i(n) -\sigma(1- \phi^i(n))-\eta= -\alpha^* \phi^i(n)-\gamma,
\end{equation}
where
\begin{eqnarray}
   \phi^i(n)&=& z \sum_{j=1}^{n} {n-1 \choose j-1} \frac{z^{n-1} (1-z)^{n-j}}{j}\nonumber\\&=&\frac{1-(1-z)^n }{n},\nonumber
\end{eqnarray}
with $z=p_cp(1-q_{\tau})$. Next, after some algebras from (\ref{eq})  we obtain
\begin{align}
\alpha^*=\frac{\lambda\sigma(n-1+(1-z^*)^{n})-n(\lambda(\gamma-e_r-e_t)-e (1-q_{\tau}))}{2\lambda(1-(1-Z^*)^{n})}.
\end{align}
where $z^*=(1-q_{\tau})p_cp^*$.\\
At Nash equilibrium, The source's objective is
\begin{equation}
 n\phi^i(n) \geq \delta,
\end{equation}
then 
\begin{equation}
n\frac{1-(1-(1-q_{\tau})p_cp^ *)^n }{n}\geq \delta,
\end{equation}
after some algebras we obtain,
\begin{equation}
 p^* \geq \frac{1- \sqrt[n]{1- \delta}}{(1-q_{\tau})p_c}
\end{equation}
\end{proof}

From now on, we refer to the case of strict equality as the Efficient Satisfaction Equilibrium (ESE). Indeed this point correspond to the minimum satisfaction level (QoS threshold) of the source node which means there will be no incentive for the source node to deviate unilaterally. This also correspond to a minimum stable accepting probability such that the whole source and relay nodes have no incentive to change their decision.

\subsection{Efficient Satisfaction Equilibria (ESE)}\label{ESE}
This section exhibits some properties of the efficient satisfaction equilibrium defined in the previous subsection. We consider and extend the Pareto-efficiency \cite{webb2007game},\cite{fudenberg1991game} as criterion to discuss and investigate its efficiency. Notice that the ESE is not always Pareto optimal.
\begin{definition}
A equilibrium is said to be strongly Pareto-optimal, if no player's payoff can be increased without decreasing the payoff of another player, i.e. $\nexists (\mbox{p'},\mbox{$\alpha$'}) $, such
  \begin{eqnarray*}
  \forall i \in n, E(\mathcal{U}_i(\mbox{p'},\mbox{$\alpha$'},n)] \geq E[\mathcal{U}_i( p^*,\alpha^*,n)] \nonumber \\ 
  \text{and} \quad \exists j \in n,  E(\mathcal{U}_j(\mbox{p'},\mbox{$\alpha$'},n)] > E[\mathcal{U}_j(p^*,\alpha^*,,n)] \nonumber \\ 
  \end{eqnarray*}
\end{definition}
\begin{proposition}
The ESE of the induced non-cooperative, symmetric One-shot caching game is Pareto optimal.
\end{proposition}
\begin{proof} In order to prove the strong Pareto-optimality, it is enough to show that for any couple $(\mbox{p'},\mbox{$\alpha$'})$ no strictly higher payoff can be obtained, without decreasing the payoff of other players. Let us, assume that $\exists (\mbox{p'},\mbox{$\alpha$'}) $ such as,
\begin{equation*}
\begin{cases}
& \mbox{p'}=\beta p^*,   \quad  \beta  >1 \\
 &\mbox{$\alpha$'}=\psi \alpha^*,  \quad   0<\psi < 1
\end{cases}
\end{equation*}
The meaning of this configuration: The source could increase its utility, if reward value decreases and the acceptance probability increases so,
  \begin{eqnarray*}
   E(\mathcal{U}_s( \mbox{p'},n)]&=&1-(1-(1-q_{\tau})p_c \mbox{p'})^n  \nonumber \\
          &=&1-(1-(1-q_{\tau})p_c \beta p^* )^n  \nonumber \\&&
   > 1-(1-(1-q_{\tau})p_c p*)^n  \nonumber \\
   \iff  E(\mathcal{U}_s( \mbox{p'},n)] &> &E(\mathcal{U}_s( p^*,n)] 
   \end{eqnarray*}
, i.e., the configuration increases the well-being of the source. However, this configuration can also decrease the well-being of the relays 
\begin{eqnarray*}
 && E(\mathcal{U}_i(\mbox{p'},\mbox{$\alpha$'},n))=\mbox{p'} D(\mbox{p'},n_a) \mathcal{U}_i(\mbox{'a'},\mbox{$\alpha$'},n_a) \nonumber \\ & &  + (1-\mbox{p'})D(\mbox{p'},n_a) \mathcal{U}_i(\mbox{'r'},\mbox{$\alpha$'},n_a) < E(\mathcal{U}_i(p^*,\alpha^*,n)) ,
\end{eqnarray*}
with $D(\mbox{p'},n_a)=\sum_{n_a=0}^{n-1} {n-1 \choose n_a}\mbox{p'}^{n_a}(1-\mbox{p'})^{n-n_a-1}$\\
We cannot improve the utility of the source without decrease the relays' utility, which contradicts the Pareto optimal definition. Consequently, the ESE is strong Pareto optimality.
\end{proof}
\section{Learning Algorithms}\label{LA}
In this, we present our proposal learning algorithms to attend the ESE discussed in previous section. In fact, we will give a formally description of the two Stochastic learning algorithm. Indeed, the stochastic learning technique has been successfully used in Distributed system, particularly in wireless networks. Briefly,  At each  iteration,  the automatons uses only the estimated value of their payoff to update their action value till converge to their unique best response. Indeed,  the source's strategy is a independent, continue action which is select a reward value $\alpha \in [0,\alpha_{max}]$, based on local observations, satisfying that the file-caring-relays succeed with probability greater then $\delta$. Precisely, the source does not need the strategies of the relays, it observes only its own payoff and its reward value assigned. Consequently, we propose a satisfactory stochastic learning algorithm based on \cite{sabir2009stochastic} which will lead the source to its optimal decision.
\begin{algorithm}
 \KwResult{Satisfaction equilibrium reward value $ \alpha^*$ }
 \textbf{Initialization}\; \par
 Assign a value for $\alpha \in [0,\alpha_{max}]$.\\
Expected  payoff  $\mu^*$.\\
\textbf{ Learning pattern:} For each iteration k\\
 Observe the estimate value of payoff $\hat{Us}$
 \[
 \hat{U_s}^{k+1}=\hat{U_s}^{k}+\epsilon_{k+1}(\mu^*-\hat{U_s}^{k})
 \]
\[
  \alpha_{k+1}= \max(\min(\alpha_{max},\alpha+\epsilon_{k+1} (\mu^*- \hat{U_s}^{k+1})),0)
\]
 \caption{Source satisfactory equilibrium stochastic learning algorithm}
\end{algorithm}\par
As regards the relays which will take the decision locally and independently based their probability distribution, according to reward value $ \hat{\alpha}$ estimated by source, so we propose imitative \textbf{CO}mbined fully \textbf{DI}stributed \textbf{PA}yoff and \textbf{S}trategy (CODIPAS) \cite{tembine2012distributed}. This choice of CODIPAS  is justified by its functionality such as, the player only need to observe the realization of their utility during previous iterations. They play independently their strategies based on the outdated observation. 
\begin{algorithm}
 \textbf{Initialization:}\; for each relay $i \in n$ do \par
 $\hat{U}_{i,0}^a,\hat{U}_{i,0}^r, p_{i,0}$\\ 
 Define the sequence up to T : $M_{i,k}^a, M_{i,k}^r, L_{i,k}^a, L_{i,k}^r$ for $ k \in \{1,....,T\}$. \\ 
\textbf{ Learning pattern:}  for each relay $i \in n$ do
 \[
 \hat{U}_{i,k+1}^a=\hat{U}_{i,k+1}^a+M_{i,k}^ae_{i,k}(U_{i,k}-\hat{U}_{i,k}^a)
 \]
  \[
 \hat{U}_{i,k+1}^r=\hat{U}_{i,k+1}^r+M_{i,k}^a(1-e_{i,k})(U_{i,k}-\hat{U}_{i,k}^r)
 \]
  \[
 p_{i,k+1}= \frac{p_{i,k}(1+L_{i,k}^a)^{\hat{U}_{i,k}^a}}{p_{i,k}(1+L_{i,k}^a)^{\hat{U}_{i,k}^a}+(1-p_{i,k})(1+L_{i,k}^r)^{\hat{U}_{i,k}^r}} 
 \]
 \caption{Imitative CODIPAS}
\end{algorithm}
where $e_k$  is the unit vector with the $ i^{th}$ component unity corresponding to the action selected at $k$

\section{Numerical Investigations}\label{MI}
We present here some numerical examples evaluating the performance of our contributions. We consider the following setting: $\delta=0.21$, $\sigma=0.2$ , $\gamma=0.15$ , $e=3.8\times10^{-5}$, $ e_r=2\times10^{-5}$ $e_t=2\times10^{-5}$, $\alpha_{max}=5, n=7$, and  we depict the behaviour of the source node and the
relay nodes while varying the horizon $\tau$ (file lifetime), the parameter $\lambda$ that stands for the contact rate and the number of the relay nodes participating in transmission. 
\begin{figure}[!htbp]
\centerline{\includegraphics[height=5cm,width=9cm]{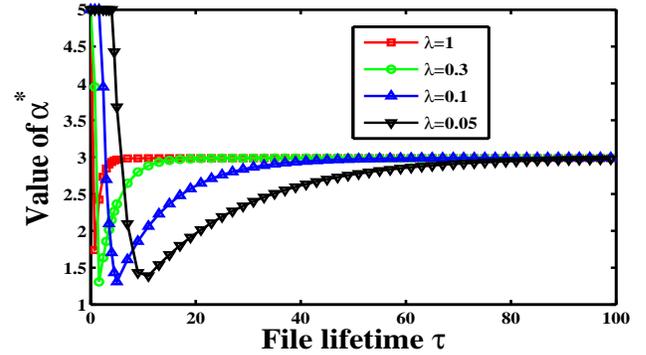}}
\centerline{(a)}
\centerline{\includegraphics[height=5cm,width=9cm]{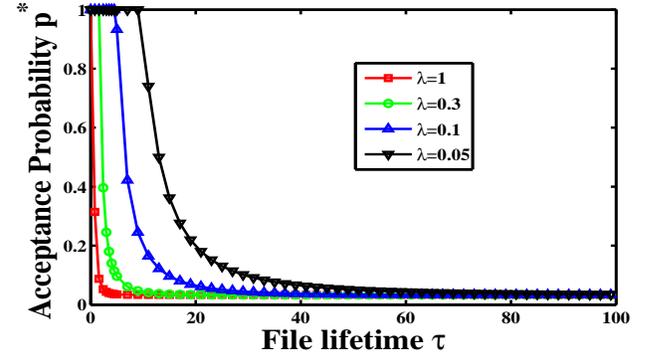}}
\centerline{(b)}
\caption{\footnotesize{Reward value and acceptance probability as function of file lifetime and contact rate}}\label{fig1}
\end{figure}
\par
We depict in Fig.\ref{fig1} (a) and Fig.\ref{fig1} (b) the acceptance probability and the reward value while varying the file lifetime for several values of the contact rate at SE. We notice that the relay nodes are cooperative and assist the file forwarding while  the source offers high value of reward. Namely, the relay nodes tend to accept with probability $1$. However, when the source decreases the reward the relay nodes decrease automatically their probability to cooperate, which is quite intuitive. This can be explained by the behaviour of the source and relays nodes, such as the source objective is offering minimum value of reward and ensuring its constraint to attain desired delivery probability $\delta$ , while the relay nodes have a problem of trade-off between the reward offered, energy and file lifetime. Indeed, the relay has benefit to accept caching the file as its expected lifetime and the probability to contact the destination are low, because the source is willing to give high value of reward. \par

Fig.\ref{fig2} (a) and Fig.\ref{fig2} (b)  show the impact of the number of relay nodes on the acceptance probability for different value of file lifetime and contact rate at SE, this probability decreases as the number of relays increases because the existence of several opponent can decrease the delivery probability of each relay nodes, hence their acceptance probability decreases afraid of to accept and fail to delivery and incur a punishment. Moreover, the influence  of $\lambda$ and $\tau$ is illustrated. Acceptance probability decreases as long as $\lambda$, $\tau$ increase, this can be easily explained by the reward offered by the source node which tries to minimum the reward value, then as $\lambda$, $\tau$ increase this value of reward is not more beneficial, so it does not cover the transmission cost of the relays. Consequently, they  relays tend  to not cooperate.

\begin{figure}[!htbp]
\centerline{\includegraphics[height=5cm,width=9cm]{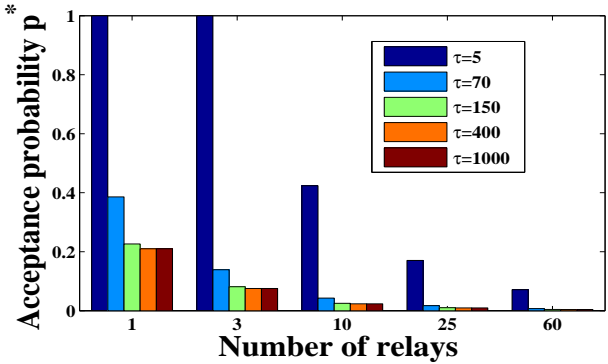}}
\centerline{(a)}
\centerline{\includegraphics[height=5cm,width=9cm]{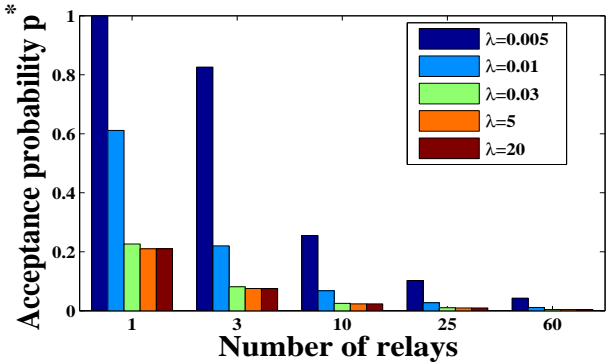}}
\centerline{(b)}
\caption{\footnotesize{ Acceptance probability as function of number of relays, file lifetime and contact rate}}\label{fig2}
\end{figure}
We depict in Fig.\ref{fig4} (a)  the delivery probability as function of the file lifetime and contact rate at SE, the figures illustrate the satisfaction regions where the source attains its objective, so while increasing the file lifetime the delivery probability increases till converges to the desired value  $\delta$. Hence, for each value of file lifetime $\lambda(\tau)$ corresponds a value of contact rate $\lambda(\tau^*)$ where the satisfaction regions of source begin. The same remark for the contact rate in Fig.\ref{fig4} (b)  for each value of file lifetime $\tau(\lambda)$ corresponds a value of contact rate $\tau(\lambda^*)$ where the satisfaction regions of source begin. \par
\begin{figure}[!htbp]
\centering
\centerline{\includegraphics[height=5cm,width=9cm]{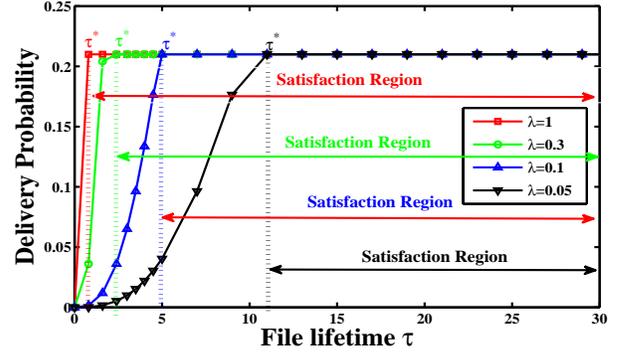}}
\centerline{(a)}
\centerline{\includegraphics[height=5cm,width=9cm]{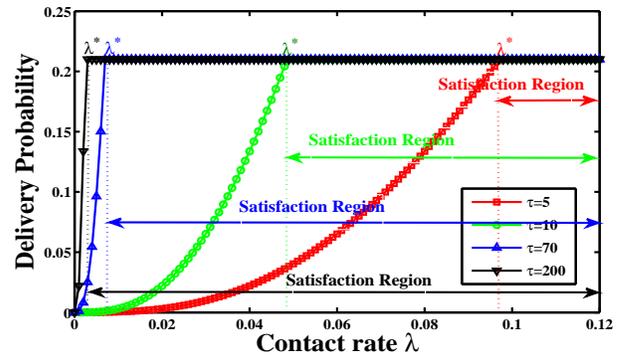}}
\centerline{(b)}
\caption{\footnotesize{ Delivery probability as function of file lifetime and contact rate}}\label{fig4}
\end{figure}
\begin{figure}[!htbp]
\centering
\centerline{\includegraphics[height=3cm,width=10cm]{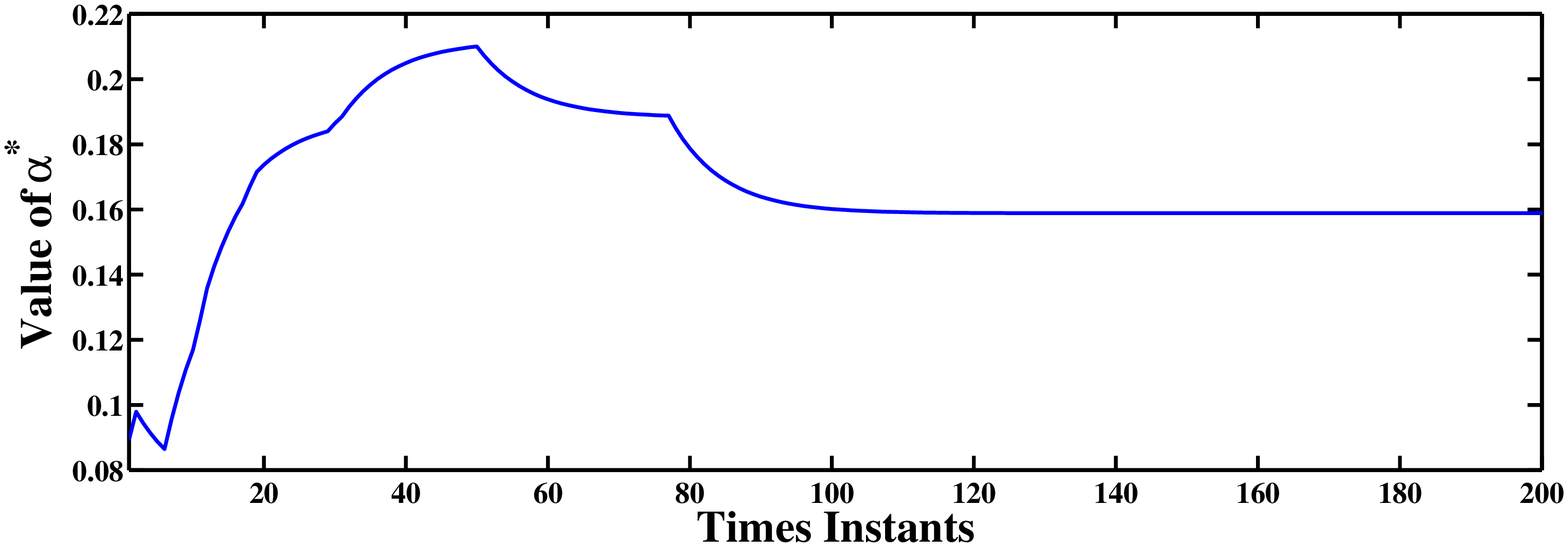}}
\centerline{(a)}
\centerline{\includegraphics[height=3cm,width=10cm]{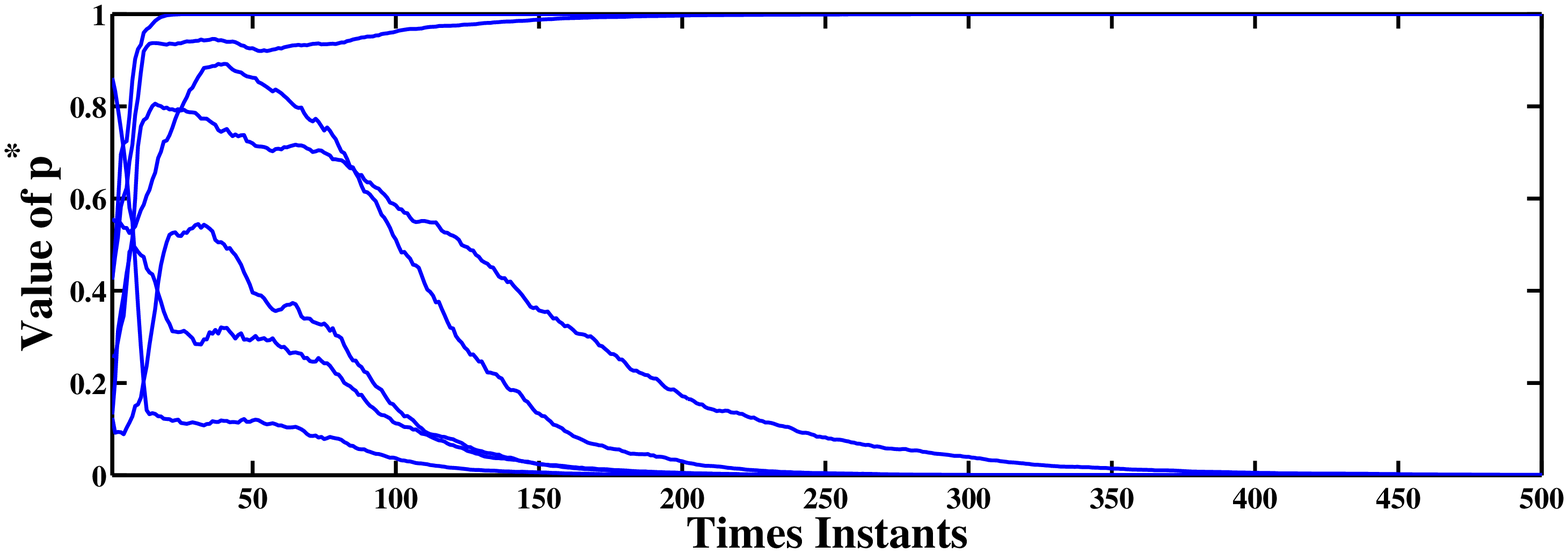}}
\centerline{(b)}
\centerline{\includegraphics[height=3cm,width=10cm]{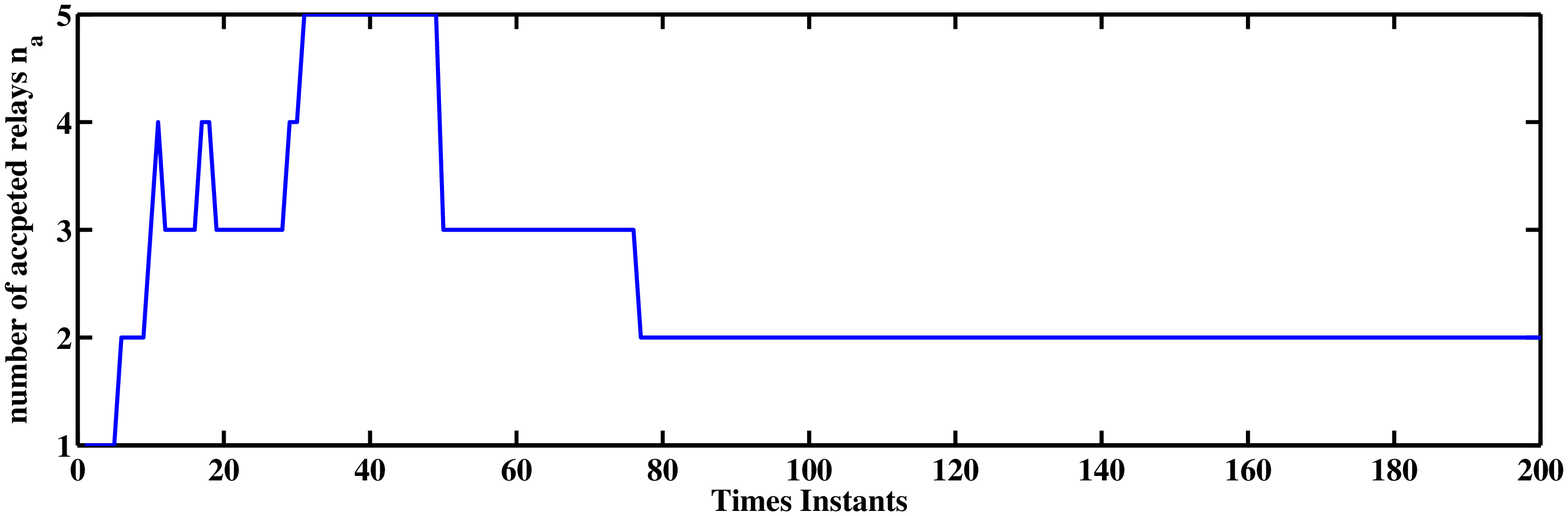}}
\centerline{(c)}
\caption{\footnotesize{Seeking the pure Satisfaction Equilibrium}}\label{fig5}
\end{figure}
In Fig.\ref{fig5} we consider a scenario involving seven relays with symmetric contact rate and file lifetime $ \lambda=0.015, \tau=100$. In fact, this figure depicts the behavior of the proposed learning algorithms over time and how they converge to the pure SE, such as Fig.\ref{fig5} (a) shows the convergence of the source satisfactory equilibrium stochastic learning algorithm to the optimal reward value  $\alpha^*$. The Fig.\ref{fig5} (b) stands for the  convergence of the relays algorithm (Imitative CODIPAS) to the pure optimal acceptance probability $p^* \in \{0,1\}$ and the Fig.\ref{fig5} (c) shows the number of relays that accept to cache, i.e the number of relay that play the pure action "accept". \par
\begin{figure}[!htbp]
\centerline{\includegraphics[height=3.2cm,width=10cm]{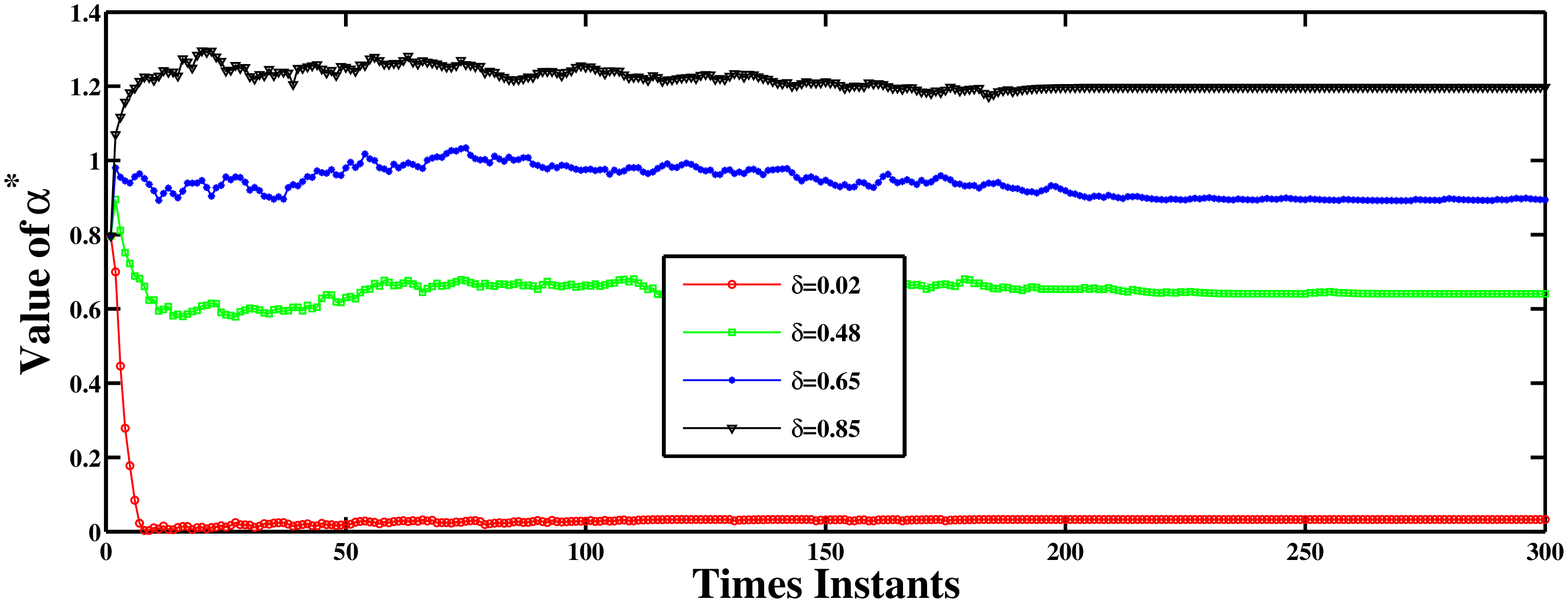}}
\centerline{(a)}
\centerline{\includegraphics[height=3.2cm,width=10cm]{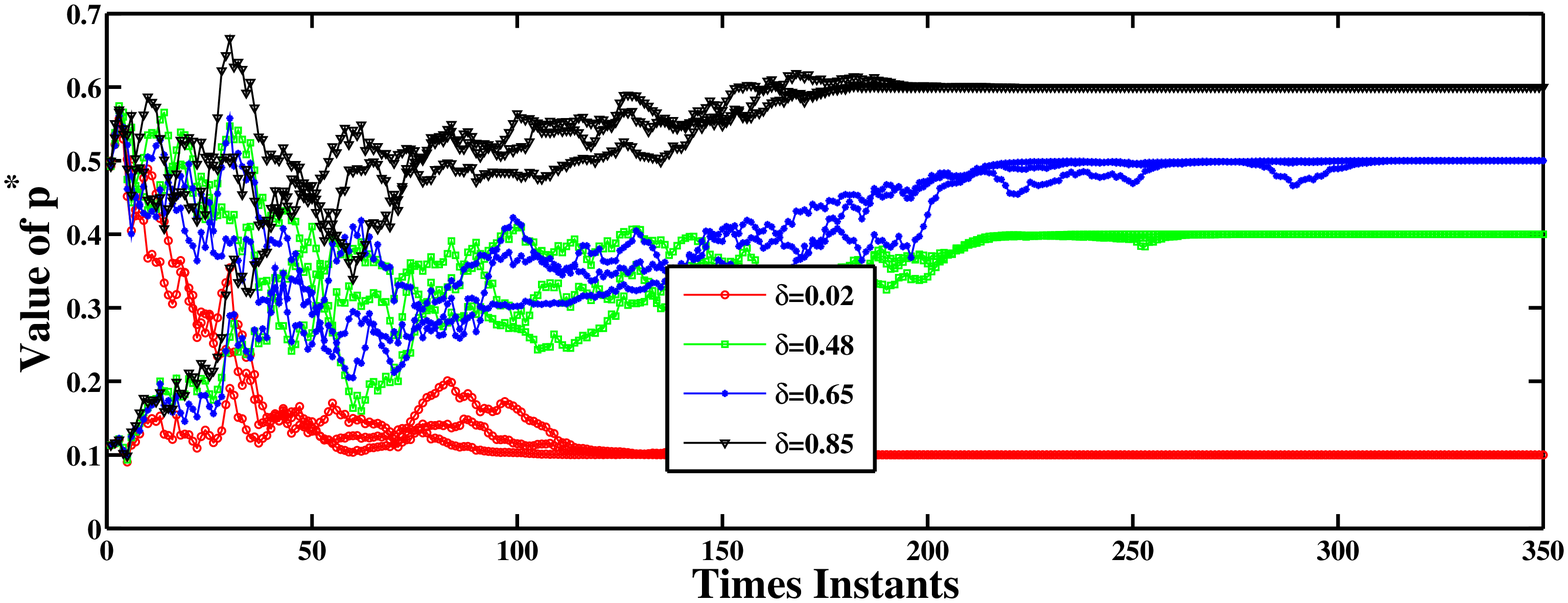}}
\centerline{(b)}
\caption{\footnotesize{Seeking the mixed Satisfaction Equilibrium}}\label{fig7}
\end{figure}
Moreover, in Fig.\ref{fig7} we considered a scenario including three relays with symmetric contact rate and file lifetime $ \lambda=0.015, \tau=100$ and  source with with different target value of delivery probability $\delta \in \{0.02, 0.48, 0.65, 0.85\}$. Hence, the Fig.\ref{fig7} (a) illustrates the convergence of the source satisfactory equilibrium stochastic learning algorithm to the optimal reward value  $\alpha^*$ that it is willing to offer  for the purpose that the relays accept to cooperate with high probability for different value of $\delta$. Whereas, the Fig.\ref{fig7} (b) shows the different acceptance probability where the relays algorithm converge for different value of reward. Precisely, the two algorithms converge independently and simultaneously to the symmetric pair vector $(\alpha^*,p^*)$ for each value of  target delivery probability desired by the source.

\section{Conclusion}\label{Con}
We investigated the support of QoS in DTN-like networks under energy/reward trade-off. We formulated the interaction between a source node and a set of relay nodes as a non-cooperative satisfaction game. Full satisfaction equilibria characterization for both pure and mixed strategies were provided. Then, we proposed two fully distributed learning algorithms to guarantee discovery of the source/relays satisfactory strategies without knowledge of any external information. The SE ensures to the source node to meet its target delivery probability while maximizing the payoff functions of the relay nodes.

\bibliographystyle{IEEEtran}
\bibliography{Ref}
\end{document}